\documentclass[12pt]{amsart}

\usepackage{amsthm,amsfonts,amsmath,amsthm, amssymb}
\usepackage{verbatim}
\usepackage{graphics,graphicx,colortbl}
\numberwithin{equation}{section}
\definecolor{Mygrey}{gray}{0.8}
\newcommand{\bea}{\begin{eqnarray}}
\newcommand{\eea}{\end{eqnarray}}
\newcommand{\be}{\begin{eqnarray*}}
\newcommand{\ee}{\end{eqnarray*}}
\newtheorem{theorem}{Theorem}[section]

\newtheorem{corollary}{Corollary}[section]

\newtheorem{example}{Example}[section]
\newtheorem{algorithm}{Algorithm}[section]

\begin{document}
\title[Boolean Networks with Single or Bistable States]{Characterization of Boolean Networks with Single or Bistable States}
\author[Yi Ming Zou]{Yi Ming Zou}
\address{Department of Mathematical Sciences, University of Wisconsin-Milwaukee, Milwaukee, WI 53201, USA} \email{ymzou@uwm.edu}
\begin{abstract}
Many biological systems, such as metabolic pathways, exhibit bistability behavior: these biological systems exhibit two distinct stable states with switching between the two stable states controlled by certain conditions. Since understanding bistability is key for understanding these biological systems, mathematical modeling of the bistability phenomenon has been at the focus of researches in quantitative and system biology. Recent study shows that Boolean networks offer relative simple mathematical models that are capable of capturing these essential information. Thus a better understanding of the Boolean networks with bistability property is desirable for both theoretical and application purposes.   
In this paper, we describe an algebraic condition for the number of stable states (fixed points) of a Boolean network based on its polynomial representation, and derive algorithms for a Boolean network to have a single stable state or two stable states. As an example, we also construct a Boolean network with exactly two stable states for the lac operon's $\beta$-galactosidase regulatory pathway when glucose is absent based on a delay differential equation model. 
\end{abstract}
\maketitle
\date{}

\section{Introduction}
\par
\par
Boolean networks were introduced in Kauffman (1969) as random models of genetic regulatory networks to study biological systems. A recent research focus of Boolean networks theory is to develop foundational theories and algorithms for Boolean networks to address questions arise from biological applications, see for example Albert and Othmer (2003), Bonneau (2008), Davidson and Levine (2008), Karlebach and Shamir (2008), Kinoshita {\it et al.} (2009), and Purnick and Weiss (2009).  To aid the study of complex biological systems, where experiments are usually expensive and time consuming, researchers use mathematical models built based on partial experimental information of these biological systems. Boolean networks offer relatively simpler such models which are capable of capturing some of the key dynamical properties, such as the stable states, of the underlying systems. As discrete time finite state dynamical systems, Boolean networks will eventually revert to certain sets of states called attractors. These attractors can be divided into two categories: stable states and cyclic states. In this paper, we restrict our attention to stable states. 
\par
The computational problem of detecting the stable states of a complex Boolean network has been discussed in several recent publications.  In Just (2006), the computational complexity of detecting the stable states of a Boolean network was discussed. In Zhang {\it et al.} (2007), an algorithm which in the average case identifies all stable states of a random Boolean network with maximum indegree $2$ in $O(1.19^n)$ time was given. In Tamura and Akutsu (2009), an algorithm for the detection of stable states in Boolean networks with non-restricted indegree that runs in $O(1.787^n)$ time was given. 
\par
There are also publications concerning the dynamics of a Boolean network that are supported by its structure, mainly on the connection between stable states and the structure of a Boolean network. In Kauffman {\it et al.} (2004), it was demonstrated that genetic networks with canalyzing Boolean rules are always stable. In Col\'{o}n-Reyes {\it et al.} (2004), the problem of when a Boolean network in which all the up-dating rules are defined by monomials has only stable states (i.e. a fixed point system) was considered. In Rudeanu (2007), a condition for a Boolean network to have only one stable state was given. Laubenbacher and Stigler (2004) and Jarrah {\it et al.} (2007) considered dynamical systems over finite fields (which includes Boolean networks as a special case) using algebraic geometry and tools in computational algebra. Choudhary {\it et al.} (2006) derived a control algorithm by minimizing a composite finite-horizon cost function that is a weighted average over the individual networks of a family of Boolean networks possessing a common set of stable states. Xiao and Dougherty (2007) studied the impact of function perturbations in Boolean networks by focusing on function perturbations in the form of a one-bit change of the truth table and investigated its impact on the stable states. Some discussions on the effects of topology on networks can be found in Nochomovitz and Li (2006), Oikonomou and Cluzel (2006), and Pomerance {\it et al.} (2009).
\par
Bistability is a fundamental phenomenon in biological systems. There have been studies devoted to the understanding of this phenomenon from a quantitative point of view at different levels. See for example, Cracium {\it et al.} (2006) and Chaves {\it et al.} (2008). An algorithm for approximating dynamical systems described by ordinary  delay differntial equations that exhibit bistability was given in Hinkelmann and Laubenbacher (2009), and their result provided further evidence that Boolean networks are capable of capturing the same key properties of biological systems as differential equations. Since circuitry properties can best be achieved by Boolean networks, one may expect Boolean networks to play an increasing role in modeling biological systems, such as genetic circuit construction (Spinzak {\it et~al}., 2005). 
\par
Our main purpose here is to provide a test for a Boolean network to have two stable states based on polynomial presentations of Boolean networks. We first describe a characterization of the number of stable states for a Boolean network by transforming the system of equations that defines the fixed points to a single equation which defines the fixed points. Then we give an algorithm for a Boolean network to have a single stable state. The bistable state algorithm will be based on the algorithm for the single stable state. Our algorithms are suitable for implementation using existing computer algebra systems such as those in MAPLE. We remark that though the detection of a stable state for a Boolean network is believed to be NP hard (Just 2006; Zhang {\it et al.}, 2007), our result here shows that the determinations of whether a Boolean network has only one or two stable states are straight forward and can be done efficiently under the assumption that the routine multiplication and addition can be performed efficiently for the Boolean polynomials that define the Boolean network under consideration. 
\par
The algorithms were tested in MAPLE using the Boolean network of Albert and Othmer (2003) on the expression pattern of the segment polarity genes in {\it Drosophila melanogaster}, and the Boolean network of Zhang {\it et al.} (2008) on the survival signaling in large granular lymphocyte leukemia. Both networks have more than two stable states. However, the stable states of the Boolean network of Zhang {\it et al.} (2008) do exhibit bistable like behavior for apoptosis, that is, one can divide the stable states ($6$ of them, see Appendix II) into two groups with one group corresponds to the ``on'' state for apoptosis and the other one corresponds to the ``off'' state. Since these and other published Boolean networks we tested do not have exactly two stable states, we also construct a Boolean network for the lac operon's $\beta$-galactosidase regulatory pathway which has exactly two stable states, one corresponds to the lac operon being at the ``on'' position and the other corresponds to the lac operon being at the ``off'' position. This captures the essential known information of the bistability behavior of the lac operon.  
\section{Main Results}
\par
A Boolean network with $n$ nodes can be given by a Boolean polynomial function
\bea\label{e1}
\mathbf{f}=(f_1,\ldots,f_n):\{0,1\}^n\rightarrow\{0,1\}^n,
\eea
where $\{0,1\}^n$ is the state space of all sequences of length $n$ formed by $0$ and $1$, and $f_1,\ldots,f_n$ are Boolean polynomials in $n$ variables $x_1,\ldots, x_n$. We will use either the logical operations {\bf OR} ($\vee$), {\bf AND} ($\wedge$), and {\bf NOT} ($\neg$), or the modulo $2$ arithmetic operations addition and multiplication to perform the calculations for Boolean variables and polynomials. The correspondence is as follows:
\be
x_i\wedge x_j = x_ix_j,\;\;
x_i\vee x_j = x_i+x_j+x_ix_j,\;\; \neg x_i = x_i+1.
\ee
To study the dynamical properties of a Boolean network, we consider the time-discrete dynamical system defined by:
\be
\mathbf{f}:(x_1(t),\ldots,x_n(t))\mapsto (x_1(t+1),\ldots,x_n(t+1)).
\ee
That is, the functions $f_i,\;1\le i\le n$, give the updating rules for the nodes, and the state of the $i$th node at time $t+1$ is given by the function value $f_i(x_1(t),\ldots,x_n(t))$. For gene regulatory networks, the variables $x_1,\ldots,x_n$ represent the genes and the functions $f_1,\ldots,f_n$ give the gene regulatory rules. If $x_i=1$ the corresponding gene is expressed ({\bf ON}), and if $x_i=0$ the
gene is not expressed ({\bf OFF}). The state space graph of a Boolean network $\mathbf{f}$ is defined to be the directed graph with the vertices given by the set $\{0,1\}^n$, and with the directed edges defined by function $\mathbf{f}$: there is a directed edge from vertex $\mathbf{v}_1$ to vertex $\mathbf{v}_2$ if the value of $\mathbf{f}$ at $\mathbf{v}_1$ is $\mathbf{v}_2$. A state $\mathbf{x} = (x_1,x_2,\ldots,x_n)\in \{0,1\}^n$ is a stable state of a Boolean network given in the form of a polynomial function as in (\ref{e1}) if it is a solution of the system of equations 
\bea\label{e2}
f_i(x_1,x_2,\ldots,x_n) = x_i,\quad 1\le i\le n.
\eea
That is, it is a {\it fixed point} of the Boolean network.
\par
We let $B = \{0,1\}$ and let $B[\mathbf{x}]$ be the set of all Boolean polynomials functions $B^n\rightarrow B^n$. Given a Boolean polynomial function $\mathbf{f}$ as in (\ref{e1}), we define the following associated Boolean polynomial
\bea\label{e3}
m_{\mathbf{f}} = 1+\prod_{i=1}^{n}(f_i+x_i+1),
\eea
and let $I_{\mathbf{f}}$ be the ideal of $B[\mathbf{x}]$ generated by $m_{\mathbf{f}}$, i.e.
\be
I_{\mathbf{f}} = (m_{\mathbf{f}}):= m_{\mathbf{f}}B[\mathbf{x}] = \{m_{\mathbf{f}}\cdot g\;|\; g\in B[\mathbf{x}]\}.
\ee
Let $[B[\mathbf{x}]:I_{\mathbf{f}}]$ be the index of $I_{\mathbf{f}}$ in $B[\mathbf{x}]$, i.e. the number of elements in the quotient $B[\mathbf{x}]/I_{\mathbf{f}}$. Now we can state the following theorem.
\begin{theorem} The number of stable states of $\mathbf{f}$ is equal to $\log_2[B[\mathbf{x}]:I_{\mathbf{f}}]$. In particular, $\mathbf{f}$ has a unique stable state if and only if $[B[\mathbf{x}]:I_{\mathbf{f}}] = 2$; and $\mathbf{f}$ has a $2$ stable states if and only if $[B[\mathbf{x}]:I_{\mathbf{f}}] = 2^2$.
\end{theorem}
\begin{proof} The system of equations given by (\ref{e2}) is equivalent to
\be
\bigvee_{i=1}^n(f_i+x_i) = 0,
\ee
which in turn is equivalent to
\be
0 &=& \neg(\neg(\bigvee_{i=1}^n(f_i+x_i))) = \neg(\bigwedge_{i=1}^n(f_i + x_i + 1))\\
{}  & =& 1 + \prod_{i=1}^n(f_i + x_i + 1)=m_{\mathbf{f}}.
\ee
Note that $I_{\mathbf{f}}$ is a subspace of the vector space $B[\mathbf{x}]$ (over $B$). Assume that $I_{\mathbf{f}}\ne B[\mathbf{x}]$, then $m_{\mathbf{f}}\ne 1$ and the equation $m_{\mathbf{f}}=0$ has at least one solution. Let $S$ be the set of solutions of $m_{\mathbf{f}}=0$. For an element $\mathbf{a} = (a_1,\ldots, a_n) \in B^n$, we define 
\bea\label{e4}
p_{\mathbf{a}} = (x_1+a_1+1)\cdots (x_n+a_n+1).
\eea
Note that $p_{\mathbf{a}}(\mathbf{a}) = 1$ and $p_{\mathbf{a}}(\mathbf{b}) = 0$ if $\mathbf{b}\ne \mathbf{a}$. Then for distinct $\mathbf{a}, \mathbf{b} \in S$, the polynomials $p_{\mathbf{a}},p_{\mathbf{b}}$, and $p_{\mathbf{a}}+p_{\mathbf{b}}$
are all not in $I_{\mathbf{f}}$. Thus the images of the set of polynomials $\{p_{\mathbf{a}}\;|\;\mathbf{a}\in S\}$ form a linearly independent subset in $B[\mathbf{x}]/I_{\mathbf{f}}$, and therefore $|S|\leq\log_2[B[\mathbf{x}]:I_{\mathbf{f}}]$. On the other hand, if $\mathbf{a}\notin S$, then $p_{\mathbf{a}}\in I_{\mathbf{f}}$. Since the set $\{p_{\mathbf{a}}\;|\;\mathbf{a}\notin S\}$ is linearly independent in $I_{\mathbf{f}}$, $\dim(I_{\mathbf{f}})\geq |S|$, so $\dim(B[\mathbf{x}]/I_{\mathbf{f}})\leq |S|$, and thus $\log_2[B[\mathbf{x}]:I_{\mathbf{f}}]\le |S|$.
\end{proof}
\par
As a consequence of the proof of the above theorem, we have the following corollary.
\begin{corollary}\label{c1}
Given a Boolean network $\mathbf{f}$ as defined by (\ref{e1}). 
\par
(1) The fixed points of $\mathbf{f}$ are the solutions of a single equation $m_{\mathbf{f}}=0$. 
\par
(2) The Boolean network $\mathbf{f}$ has exactly the elements in $S\subseteq B^n$ as its stable states if and only if 
\be
m_{\mathbf{f}}=1+\sum_{\mathbf{a}\in S}p_{\mathbf{a}}.
\ee
\end{corollary}
\par
This leads to the following single stable state Boolean network testing algorithm:
\par
\begin{algorithm}\label{a1} 
\par
Single Stable State Boolean Network Test.
\par\medskip
{\bf INPUT:} $\mathbf{f}=(f_1,f_2,\ldots,f_n),\; f_i\in B[\mathbf{x}],\; 1\le i\le n$.
\par
{\bf OUTPUT:} YES or NO; if YES also returns the single stable state of $\mathbf{f}$.
\par
\medskip
1. For each $1\le i\le n$, determine if the Boolean polynomial $1+m_{\mathbf{f}}$ has the factor $x_i$ or not. Let $I$ be the subset of $\{1,2,\ldots, n\}$ such that the corresponding answers are yes.
\par
2. For $\forall j\notin I$, determine if the Boolean polynomial $1+m_{\mathbf{f}}$ has the factor $x_i+1$ or not.
\par
3. If all the answers in step 2 are yes, then return ``YES'' and the single stable state $\mathbf{a} = (a_1,a_2,\ldots, a_n)$, where
\be
a_i = \left\{ \begin{array}{rll}
              1 & \mbox{if} & i\in I;\\
              0 & \mbox{if} & i\notin I.
              \end{array}\right.
\ee
Otherwise, return ``NO''.
\end{algorithm}
\begin{proof}
By Corollary \ref{c1}, $\mathbf{f}$ has a single stable state is equivalent to the condition $m_{\mathbf{f}}=1+p_{\mathbf{a}}$ for some $\mathbf{a}\in B^n$. That is
\bea\label{e5}
\prod_{i=1}^{n}(f_i+x_i+1) = p_{\mathbf{a}}\;\; \mbox{for some $\mathbf{a}\in B^n$}.
\eea
Note that in any case, the element $p_{\mathbf{a}}$ is of the form 
\bea\label{e6}
x_{i_1}\cdots x_{i_s}(x_{i_{s+1}}+1)\cdots (x_{i_{n}}+1),
\eea
where $(i_1i_2\cdots i_n)$ is a permutation of $(12\cdots n)$. So we can first find out which of $x_i$ ($1\le i\le n$) is a factor of the left hand side of equation (\ref{e5}); then for those $x_i$ which are not factors, determine whether $x_i+1$ is a factor instead. This procedure will decide if the left hand side of equation (\ref{e5}) is the form as in (\ref{e6}). To finish the proof, note that the solution of $1+p_{\mathbf{a}} = 0$, when $p_{\mathbf{a}}$ is written as in (\ref{e6}), is given by $x_{i_j} = 1$ for $1\le j\le s$ and $x_{i_j} = 0$ for $s+1\le j\le n$.  
\end{proof}
\par
We remark that it is clear that the computation needed to run this algorithm is the determination whether the variables $x_i$ or $x_i+1$ are factors of the Boolean polynomial $1+m_{\mathbf{f}}$ or not. For $x_i$, this can be done by inspecting the terms. For $x_i+1$, one can make a substitution by substitute in $x_i+1$ for $x_i$ and thus transferring the problem to the determination of whether $x_i$ is a factor or not. Clearly, this can be done efficiently if the computation of the polynomial $1+m_{\mathbf{f}}$ can be done efficiently. We also want to remark that by comparing with the result of Rudeanu (2007), it is clear that the conditions in Algorithm \ref{a1} are simpler and easier to verify than the conditions provided in Rudeanu (2007) (see also a counting of the unique fixed point systems for the case $n=2$ in section 3 below). A simple extension of Algorithm (\ref{a1}) can be used to determine if a Boolean network has exactly two stable states or not. 
\par
\begin{algorithm}\label{a2} Bistable State Boolean Network Test.
\par\medskip
{\bf INPUT:} $\mathbf{f}=(f_1,f_2,\ldots,f_n),\; f_i\in B[\mathbf{x}],\; 1\le i\le n$.
\par
{\bf OUTPUT:} YES or NO; if YES also returns the two stable states of $\mathbf{f}$.
\par\medskip
1. Apply Algorithm \ref{a1} to $1+m_{\mathbf{f}}$ to find the possible factors of the forms $x_i$ or $x_i+1$. If $1+m_{\mathbf{f}}$ is a product of these forms, then return ``NO'' and terminate the procedure, since the Boolean network has a single stable state. Otherwise, assume that 
\bea\label{e7}
1+m_{\mathbf{f}} = x_{i_1}\cdots x_{i_s}(x_{j_1}+1)\cdots (x_{j_t}+1)g,
\eea 
where $g\ne 1$ is an element in $B[\mathbf{x}]$ such that if 
\be
i\notin I:=\{i_1,\ldots, i_s, j_1,\ldots, j_t\}
\ee
 then
\be
g(x_i = 0) \ne 0,\;\; g(x_i=1)\ne 0.
\ee
\par\medskip
2. Pick any $i\notin I$, and put $c_i = 1,\;d_i = 0$. Apply Algorithm \ref{a1} to 
\bea\label{e8}
g_0 : = g(x_i=0)\;\;\mbox{and}\;\; g_1:=g(x_i=1),
\eea
with the set of variables $\{x_j\;|\;j\notin I\;\mbox{and}\;j\ne i\}$. If
\bea\label{e9}
g_0 = \prod_{\substack{j\notin I\\ j\ne i}}(x_j+c_j)\;\;\mbox{and}\;\;g_1 = \prod_{\substack{j\notin I\\ j\ne i}}(x_j+d_j)
\eea
such that $c_j+d_j = 1,\;\forall j\notin I$ and $j\ne i$, then return ``YES'' and the two stable states $\mathbf{a} = (a_1,\ldots,a_n)$ and $\mathbf{b}=(b_1,\ldots, b_n)$ of $\mathbf{f}$, where 
\begin{gather*}
a_{i_k} = 1,\;1\le k\le s;\;\; a_{j_{k}} = 0,\; 1\le k\le t;\;\; \\
a_j = c_j +1,\; j\notin I;
\end{gather*}
and 
\begin{gather*}
b_{i_k} = 1,\;1\le k\le s;\;\; b_{j_{k}} = 0,\; 1\le k\le t;\;\;\\
 b_j = d_j +1,\; j\notin I.
\end{gather*} 
Otherwise, return ``NO''.
\end{algorithm}
\par
\begin{proof} By Corollary \ref{c1}, the fact that the Boolean network $\mathbf{f}$ has exactly two stable states is equivalent to the equation $m_{\mathbf{f}} = 1+p_{\mathbf{a}}+p_{\mathbf{b}}$ for two distinct elements $\mathbf{a}, \mathbf{b}\in B^n$, where $p_{\mathbf{a}}$ and $p_{\mathbf{b}}$ are defined by (\ref{e4}). If $m_{\mathbf{f}} = 1+p_{\mathbf{a}}+p_{\mathbf{b}}$, after factoring out the common factors as in (\ref{e7}), we have 
{\small
\be
p_{\mathbf{a}}+p_{\mathbf{b}} &=& x_{i_1}\cdots x_{i_s}(x_{j_1}+1)\cdots (x_{j_t}+1)\times\\
                     {}       &{}& \quad (\prod_{j\notin I}(x_j+c_j) + \prod_{j\notin I}(x_j+d_j)),
\ee}
where $c_j+d_j = 1,\;\forall j\notin I$. This implies that condition (\ref{e9}) is necessary. Suppose that the algorithm returns ``YES''. We note that according to the definition of $g_0$, we have
\bea\label{e10}
g &=& x_ih+g_0,\;\;\mbox{for some $h\in B[\mathbf{x}]$}\\
{}&=& x_i(h+g_0) + (x_i+1)g_0. \nonumber
\eea
Then according to the definition of $g_1$, by setting $x_i = 1$ in (\ref{e10}), we have $h+g_0 = g_1$. Thus we must have $m_{\mathbf{f}} = 1+p_{\mathbf{a}}+p_{\mathbf{b}}$ for two distinct elements $\mathbf{a}, \mathbf{b}\in B^n$.
\end{proof}
\section{Counting and Examples}
\par
In this section, we use two simple examples to explain the two algorithms given in the previous section. Before giving these examples, we want to consider the counting problem: How many Boolean networks with a fix number of stable states are there? This question can be answered using the formula given on p. 176 of the book by Goulden and Jackson (1983). Since our definition allows a Boolean network with a fix number of  stable states to have other cyclic states, we can find the number of Boolean networks with $n$ nodes that have either a single stable state or two stable states as follows. Let
\be
\Lambda_1 = \{ \mathbf{i} = (i_1,i_2,\ldots)\;|\; 1+ 2i_1+\cdots \le 2^n\}.
\ee
Then the number of Boolean network with $n$ nodes that have a single stable state is given by
\bea\label{e11}
\sum_{\mathbf{i}\in \Lambda_1}\frac{1}{(2^{i_1}i_1!)\cdots}(2^n)^{2^n-q}q!\left(\begin{array}{c} 2^n-1\\ q-1  \end{array}\right),
\eea
where $q = 1 + 2i_1 +\cdots\le 2^n$. We remark that the leading $1$ on the right hand side of the equation for $q$ signifies that there is one stable state, the $2i_1$ signifies that there are $i_1$ cycles of length $2$, etc. 
\par
For example, if $n=2$, then since $2^2 = 4$, there are $3$ possible values for $q$:
\be
\begin{array}{ccll}
q &=& 1, &\mbox{(one stable state, no other cycle)}\\
q &=& 1+2\cdot 1 = 3, &\mbox{(one stable state, one cycle of length $2$)}\\
q &=& 1+3\cdot 1 = 4. &\mbox{(one stable state, one cycle of length $3$)}
\end{array} 
\ee
Using formula (\ref{e11}), we find that the number of Boolean networks in two variables with a unique stable state is (compare with Rudeanu (2007), where $36$ such fixed point systems were claimed):
\be
4^3+\frac{4^{4-3}}{2}3!\left(\begin{array}{c} 3\\ 2  \end{array}\right)
+\frac{4^{4-4}}{3}4!\left(\begin{array}{c} 3\\ 3  \end{array}\right) = 108.
\ee
Among these are the $4^3=64$ Boolean networks without any cycle of length $> 1$.
\par
The number of Boolean network with $n$ nodes that have two stable states is given by
\bea\label{e12}
\sum_{\mathbf{i}\in \Lambda_2}\frac{1}{(2!)(2^{i_1}i_1!)\cdots}(2^n)^{2^n-q}q!\left(\begin{array}{c} 2^n-1\\ q-1  \end{array}\right),
\eea
where 
\be
\Lambda_2 = \{ \mathbf{i} = (i_1,i_2,\ldots)\;|\; 2+ 2i_1+\cdots \le 2^n\},
\ee
and $q = 2 + 2i_1 +\cdots\le 2^n$. In particular, there are $54$ Boolean networks in two variables that have exactly two stable states.
\par 
\medskip
For a discussion of constructing Boolean networks with prescribed attractor structures, we refer the reader to Pal {\it et al.} (2005).
\par\medskip
We now consider two examples. The first example is a Boolean network with a single stable state.
\begin{example} Consider a Boolean network with $3$ nodes defined by
{\small
\be
f_1 &=& x_1 + x_2 + x_3 + x_1x_2x_3,\\
 f_2 &=&  x_1(1 + x_3 + x_2x_3),\\
    f_3 &=&  x_2+x_3+x_1x_3.
\ee}
Then we can easily factor
\be
1+m_{\mathbf{f}} = (x_1+1)(x_2+1)(x_3+1).
\ee
So this Boolean network has a single stable state $000$.
\end{example}
\par
The second example is a Boolean network with more than two stable states.
\begin{example} Let the Boolean network $\mathbf{f}$ with $6$ nodes be defined by
{\small
\be
  f_1 &=& x_1(x_2+1)(x_4+1),\\ 
  f_2 &=& (x_1+1)x_2(x_3+1),\\
  f_3 &=& x_1+x_3+x_1x_3,\\
  f_4 &=& x_2+x_4+x_2x_4,\\ 
  f_5 &=& (x_2+1)(x_4+1)+x_5(x_1+1)(x_3+1)\\
   {} &{}& \quad +x_5(x_1+1)(x_2+1)(x_3+1)(x_4+1),\\
  f_6 &=& (x_1+1)(x_3+1)+x_6(x_2+1)(x_4+1)\\
   {} &{}& \quad +x_6(x_1+1)(x_2+1)(x_3+1)(x_4+1).
\ee}
Apply Algorithm (\ref{a1}), we find that $1+m_{\mathbf{f}}$ has no common factor, so we just let $g = 1+m_{\mathbf{f}}$ and pick $i=1$ to define $g_0$ and $g_1$ as in Algorithm (\ref{a2}):
{\small
\be
g0 &=& g(x_1 = 0) = x_5x_2x_6+x_5x_4x_6+x_5x_2x_6x_3+x_6x_4\\
           {} &{}& \quad +x_6x_2x_3x_4+x_2x_3x_4+ x_5x_2x_3+x_5x_2x_6x_4\\
        {} &{}& \quad +x_5x_6+x_5x_6x_3+x_5x_3+x_3x_4,\\
g1 &=& g(x_1 = 1) = x_5x_2x_3+x_5x_2x_3x_4+x_5x_3x_4+x_5x_3.
\ee}
Apply Algorithm (\ref{a1}) to both $g0$ and $g1$, we find that $g1$ factors as $x_3x_5(x_4+1)(x_2+1)$ while $g0$ does not factor. Thus by Algorithm (\ref{a2}), $\mathbf{f}$ has more than two stable states.
\end{example}
\par
We also tested our algorithms on some published Boolean networks. Two of these are given in the Appendix II. One is the Boolean network of Albert and Othmer (2003) on the expression pattern of the segment polarity genes in {\it Drosophila melanogaster}, the other is the Boolean network of Zhang {\it et al.} (2008) on the survival signaling in large granular lymphocyte leukemia. The former has $21$ nodes, and the latter has $29$ nodes. More detail can found in the Appendix. However, all these Boolean networks have more than two stable states. In the next section, we construct a Boolean network for the lac operon with exactly two stable states as an example.
\par\medskip
\section{A Boolean Network for the Lac Operon}
\par
In this section, we construct a Boolean network for the lac operon and show that it captures the bistability of the lac operon.
\par
The concept of operon was introduced by Jacob {\it et al.} (1960). An operon is made up of several structural genes arranged under a common promoter and regulated by a common operator. It is defined as a set of adjacent structural genes and the adjacent regulatory signals that affect the transcription of the structural genes. The lac operon, one of the most extensively studied operons, controls the metabolism of lactose in Escherichia coli, and it is regulated by several factors including the availability of glucose and lactose. The cell can use lactose as an energy source by producing the enzyme $\beta$-galactosidase to digest lactose into glucose. If there is glucose or there is no lactose present, then it does not produce $\beta$-galactosidase.
\par
The lac operon consists of a small promoter-operator region and three adjacent larger structural genes lacZ, lacY, and lacA. Preceding the lac operon is a regulatory gene lacI that is responsible for producing a repressor protein. In the absence of allolactose (A) the repressor binds to the operator region and prevents the RNA polymerase from transcribing the structural genes (top diagram in Fig. 1). In the presence of lactose (L), lactose is first transported into the cell by the permease and then broken down into glucose, galactose, and allolactose by $\beta$-galactosidase. When allolactose presents, a complex is formed between allolactose and the repressor that makes binding of the repressor to the operator region impossible, then the RNA polymerase bound to the promoter is able to initiate transcription of the structural genes to produce mRNA (M). The lacZ gene encodes the portion of the mRNA that is responsible for the production of $\beta$-galactosidase (B), the lacY gene produces the section of mRNA responsible for the production of a permease (P), and the lacA gene encodes the section of mRNA responsible for the production of thiogalactoside transacetylase (T) which does not play a role in the regulation of the lac operon (bottom diagram in Fig. 1).
\par
\begin{figure}[h]
\begin{center}
\includegraphics[width=3.4in, height=2.5in]{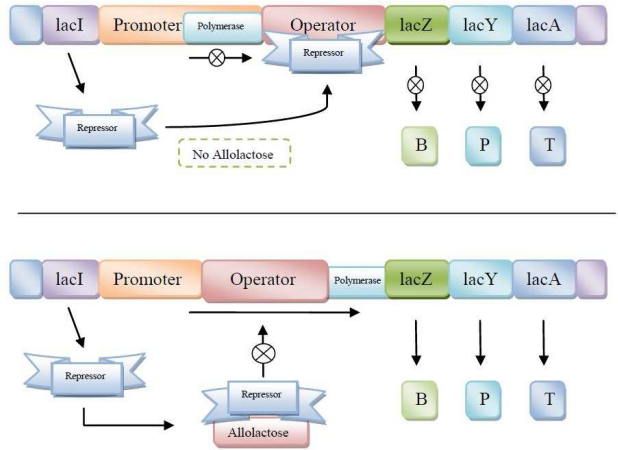} 
\caption{{\footnotesize A schematic representation of the lac operon's $\beta$-galactosidase regulatory pathway when glucose is absence. Top diagram corresponds to the case when allolactose is absence, bottom diagram corresponds to the case when allolactose is presence.}} \label{Fig1}
\end{center}
\end{figure}
\par
Experiments show that the lac operon regulatory pathway is capable of showing bistable behavior (Novik and Wiene, 1957), and mathematical modelings show that perhaps the $\beta$-galactosidase regulatory pathway is the most essential of the regulatory mechanisms in the lac operon (Yildirim and  Mackey, 2003; and Yildirim {\it et. al.}, 2004). The Boolean network constructed here for the lac operon is based on the delay differential equation model of Yildirim {\it et. al.} (2004), which is a reduced version of the model in Yildirim and Mackey (2003).  A Boolean network based on the model of Yildirim and Mackey (2003) was constructed by Hinkelmann and Laubenbacher (2009). We remark that to construct a Boolean network for the lac operon, we can start from the experimental data. However, there are many Boolean networks fit the same set of data, and it is not clear how some of these Boolean models related to the existing differential equation models. Though there may be other equally valid differential equation models correspond to these Boolean models, as an example, it seems natural to construct the Boolean model this way for, say, comparison purpose. Other simplified Boolean models were also constructed in existing literature. For example, Stigler and Veliz-Cuba (2008), Laubenbacher and Sturmfels (2009). These models are different from the model considered here, they do not base on the model of Yildirim {\it et. al.} (2004).
\par
The model of Yildirim {\it et. al.} (2004) consists of three ordinary delay differential equations. It models the concentration of mRNA, the concentration of $\beta$-galactosidase, and the concentration of allolactose:
\bea\label{e13}
\frac{dA}{dt} &=& \alpha_AB\frac{L}{K_L+L} -\beta_AB\frac{A}{K_A+A}-\widetilde{\gamma_A}A, \\ \nonumber
\frac{dM}{dt} &=& \alpha_M\frac{1+K_1(e^{-\mu\tau_M}A_{\tau_M})^n}{K+K_1(e^{-\mu\tau_M}A_{\tau_M})^n}-\widetilde{\gamma_M}M,\\ \nonumber
\frac{dB}{dt} &=& \alpha_Be^{-\mu\tau_B}M_{\tau_B} - \widetilde{\gamma_B}B,
\eea
where the terms with delays are given by
\be
A_{\tau_M} := A(t-\tau_M),\quad M_{\tau_B} := M(t-\tau_B).
\ee
The assumption is that lactose is in a quasi-steady state and hence can be considered as a constant. We now use the method developed in Hinkelmann and Laubenbacher (2009) to construct a Boolean network based on system (\ref{e13}). We introduce new variables $A_1$, $A_{old}$, $M_1$, and $M_{old}$ in order to capture the delay effects correspond to $A_{\tau_M}$ and $M_{\tau_B}$. We assign the variables $x_1,\ldots,x_7$ by
\be
(A, A_1, A_{old}, B, M, M_{1}, M_{old}) = (x_1,x_2,x_3,x_4,x_5,x_6,x_7).
\ee
Then we derive the following Boolean network with $7$ variables (see Appendix I for more detail):
\bea\label{e14}
f_1 &=& x_2x_4+x_1x_2(x_3+1)(x_4+1),\\ \nonumber
f_2 &=& x_1+x_2+x_1x_2,\quad f_3 = (x_2+1)x_4,\\ \nonumber
f_4 &=&  x_4+x_6+x_7+x_4x_6+x_4x_7+x_6x_7+x_4x_6x_7,\\ \nonumber
f_5 &=& x_2 + (x_2+1)x_5(x_7+1),\; f_6 = x_5,\; f_7 = x_2+1.
\eea
Apply Algorithm (\ref{a2}), we find that {\small
\be
1+m_{\mathbf{f}} &=& (x_1+1)(x_2+1)x_3x_4(x_5+1)(x_6+1)x_7\\
        {}       &{}& \qquad + \;x_1x_2(x_3+1)x_4x_5x_6(x_7+1).
\ee}
Thus this Boolean network has two stable states $(0 0 1 1 0 0 1)$ and $(1 1 0 1 1 1 0)$. Analysis of this Boolean network (for this Boolean network, it can be done using the software DVD developed by the Applied Discrete Mathematics Group at Virginia Bioinformatics Institute\footnote{DVD available at: http://dvd.vbi.vt.edu/}) shows that all initial states corresponding to low level of allolactose (both $A$ and $A_1$ are $0$) converge to the first stable state, which corresponds to the lac operon being at the ``off'' state; and that all the initial states corresponding to high level of allolactose (at least one of $A$ and $A_1$ is $1$) converges to the second stable state, which corresponds to the lac operon being at the ``on'' state.
\section{Concluding Remarks}
\par
\par
We have formulated a characterization for the number of stable states of a Boolean network based on polynomial presentations of Boolean networks. We have also proposed two algorithms for testing if a Boolean network is a single stable state or a bi-stable state Boolean network. Our algorithms are suitable for implementation using computational algebra softwares such as MAPLE and SINGULAR. The basic method of finding the stable states of a Boolean network is to solve the associated polynomial system (\ref{e2}), using tools such as Gr\"{o}bner bases to increase the computation efficiency. The advantage of our algorithms is that they do not involve solving a multi-variable Boolean polynomial system. Rather, the efficiency of our algorithms here depends on the efficiency of the routine operations with a Boolean polynomial system. Since there are $2^n$ linearly independent monomials in $B[\mathbf{x}]$, a single polynomial can contain a large number of terms. So we see that in general, if $n$ is large, then even routine computations with these polynomials, such as multiplication and addition, can be complex. If these operations can be handled efficiently for a given Boolean network, our algorithms are then straight forward. For a recent work that aims at the computations of Boolean polynomials, we refer the reader to Brickenstein and Dreyer (2009).
\par
We also considered a Boolean network for the $\beta$-galactosidase regulatory pathway of the lac operon as an example. This switch like Boolean network is based on the delay differential equation model of Yildirim {\it et. al.} (2004). It has two stable states, all the initial states corresponding to the absence of allolactose converge to the ``off'' stable state, and all the initial states corresponding to the presence of allolactose converge to the ``on'' stable state. This example shows some of the advantages of Boolean network models: switch like properties can be achieved more easily via Boolean networks while keeping the same essential information encoded in a differential equation model. This is consistent with the finding of Hinkelmann and Laubenbacher (2009).
\par
Finally, we want to make some comments on using the polynomial $m_{\mathbf{f}}$ defined by (\ref{e3}) to find the fixed points of a Boolean network $\mathbf{f}$. Recall that all fixed points of the Boolean network $\mathbf{f}$ are given by the solutions of $m_{\mathbf{f}} = 0$ (Cor. \ref{c1}). We observe that any method of solving a system of $n$ equations in $n$ variables will provide a method for solving the single equation $m_{\mathbf{f}}= 0$, since we can form a system of $n$ equations with one of them being $m_{\mathbf{f}} = 0$ and the others being trivial (identities). Using MAPLE $11$ on a Dell laptop with the system: Intel(R)Core(TM)2 Duo CPU T9900@3.06GHz with 3.5 GB RAM, we detected $176$ fixed points for the Boolean network of Albert and Othmer (2003) in 29.28 seconds by solving the single equation $m_{\mathbf{f}}=0$, in which, the computation of $m_{\mathbf{f}}$ from $\mathbf{f}$ used $0.45$ second. For the Boolean network of Zhang {\it et al.} (2008), the computation took $21.08$ second, and $6$ fixed points were detected. More detail of these examples are provided in the Appendix II.





%
%

\section*{Appendix I}
We provide some detail for the deriving of the Boolean network (\ref{e14}) in this appendix.
\par
\subsection*{Allolactose}
\par
The gain of allolactose is due to the conversion of lactose mediated by $\beta$-galactosidase. The loss of allolactose is due to its conversion to glucose and galactose (also mediated by $\beta$-galactosidase), the degradation, and the dilution.  Since for this model, the assumption is that lactose  L is in a quasi-steady state, we thus put
\be
f_{A} = (x_{A_1}\wedge x_B)\vee (x_A\wedge x_{A_1}\wedge\neg x_{A_{old}}).
\ee
The term $x_{A_1}\wedge x_B$ corresponds to the first two terms on the right hand side of the differential equation for $A$ in (\ref{e13}), and the term $x_A\wedge x_{A_1}\wedge\neg x_{A_{old}}$ accounts for the last term. Since the change of $A_1$ depends only on $A$, we put
\be
f_{A_1} = x_A\vee x_{A_1}.
\ee
For $A_{old}$, we put
\be
f_{A_{old}} = \neg x_{A_1}\wedge x_B.
\ee
Substitute in the variables $x_1,\ldots,x_4$, we obtain the first three equations in (\ref{e14}). 
\subsection*{mRNA}
\par
The production of mRNA from DNA through transcription is not an instantaneous process. It requires a period of time for RNA polymerase to transcribe the ribosome binding site. Thus we put
\be
f_M &=& x_{A_1}\vee (x_{A_{old}}\wedge \neg x_{M_{old}}),\\
f_{M_1} &=&  x_M,\quad f_{M_{old}} = \neg x_{A_1}.
\ee
In the equation for $f_M$, the two terms correspond to the two terms on the right hand side of the differential equation for $M$ in (\ref{e13}). Substitute in the variables, we get the equations for $f_5$, $f_6$, and $f_7$ in (\ref{e14}).
\subsection*{$\beta$-Galactosidase}
\par
The production of $\beta$-galactosidase through translation of the mRNA also requires a period of time delay, so we put
\be
f_B = x_{M_1} \vee x_{M_{old}}\vee x_B,
\ee
where the first two terms correspond to the first term on the right hand side of the differential equation for $B$ in (\ref{e13}), and the last term accounts for the second term. Substitute in the variables, we get the equation for $f_4$ in (\ref{e14}).
\section*{Appendix II}
\subsection*{The Boolean network of Albert and Othmer (2003)}
\par
We introduce the variables as follows (we refer the reader to Albert and Othmer (2003) for the original Boolean network):
\par
\begin{table}[h]
  \begin{center}
  {\renewcommand{\arraystretch}{1.4}\small 
  \begin{tabular}{|c|c|c|c|c|c|c|c|}
    \hline
    SLP & $wg$ & WG & $en$ & EN & $hh$& HH \\\hline
    $x_1$ & $x_2$ & $x_3$ & $x_4$ & $x_5$ & $x_6$ & $x_7$ \\\hline \hline
    $ptc$ & PTC & PH & SMO & $ci$ & CI & CIA\\ \hline
    $x_8$ & $x_9$ & $x_{10}$ & $x_{11}$ & $x_{12}$ & $x_{13}$ & $x_{14}$\\ \hline\hline
    CIR & WG$_{i-1}$ & WG$_{i+1}$ & HH$_{i-1}$ & HH$_{i+1}$ & $hh_{i-1}$ & $hh_{i+1}$ \\\hline
    $x_{15}$ & $x_{16}$ & $x_{17}$ & $x_{18}$ & $x_{19}$ & $x_{20}$ & $x_{21}$ \\\hline 
 \end{tabular}
  }
  \label{tab: pairing1}
  \vskip 0.5cm
  \caption[Legend of variable names.]{\footnotesize Legend of variable names.}
    \end{center}
\end{table} 
\par
Then the Boolean network is given by the following polynomial functions:
\begin{align*}
f_1&=x_1,\\
f_2&=(x_{15}+1)(x_1(x_2+x_{14})+x_2x_{14}),\\
f_3&=x_2,\\
f_4&=x_1(x_{16}(x_{17}+1)+x_{17})+x_{16}(x_{17}+1)+x_{17},\\
f_5&=x_4,\\
f_6&=x_5(x_{15}+1),\\
f_7&=x_6,\\
f_8&=(x_4+1)x_{13}((x_{11}+1)(x_{20}(x_{21}+1)+x_{21})+x_{11}),\\
f_9&=(x_8+1)x_9(x_{18}+1)(x_{19}+1)+x_8,\\
f_{10}&=((x_8+1)x_9(x_{18}+1)(x_{19}+1)+x_8)(x_{20}(x_{21}+1)+x_{21}),\\
f_{11}&=f_{9}+f_{10}+1,\\
f_{12}&=x_5+1,\\
f_{13}&=x_{12},\\
f_{14}&=x_{13}((x_{11}+1)(x_{21}+1)(x_{20}+1)+1),\\
f_{15}&=f_{14}+x_{13},\\
f_{i}&=x_i\ \text{for }16\le i\le 21.
\end{align*}
\par
\medskip
Then $1+m_{\mathbf{f}}$ is
{\footnotesize
\be
{} &{}& ((x_{15}+1)*(x_1*(x_2+x_{14})+x_2*x_{14})+x_2+1)*(x_2+x_3+1)\\
{} &{}&*(x_1*(x_{16}*(x_{17}+1)+x_{17})+x_{16}*(x_{17}+1)+x_{17}+x_4+1)\\
{} &{}&*(x_4+x_5+1)*(x_5*(x_{15}+1)+x_6+1)*(x_6+x_7+1)\\
{} &{}&*((x_4+1)*x_{13}*((x_{11}+1)*(x_{20}*(x_{21}+1)+x_{21})+x_{11})+x_8+1)\\
{} &{}&*((x_8+1)*x_9*(x_{18}+1)*(x_{19}+1)+x_8+x_9+1)\\
{} &{}&*(((x_8+1)*x_9*(x_{18}+1)*(x_{19}+1)+x_8)*(x_{20}*(x_{21}+1)+x_{21})+x_{10}+1)\\
{} &{}&*((x_8+1)*x_9*(x_{18}+1)*(x_{19}+1)+x_8+((x_8+1)*x_9*(x_{18}+1)\\
{} &{}&*(x_{19}+1)+x_8)*(x_{20}*(x_{21}+1)+x_{21})+x_{11}+1)*(x_5+x_{12})\\
{} &{}&*(x_{12}+x_{13}+1)*(x_{13}*((x_{11}+1)*(x_{21}+1)*(x_{20}+1)+1)+x_{14}+1)\\
{} &{}&*(x_{13}*((x_{11}+1)*(x_{21}+1)*(x_{20}+1)+1)+x_{13}+x_{15}+1).
\ee}
\par
\medskip
Both Algorithm (\ref{a1}) and Algorithm (\ref{a2}) fail, so the Boolean network has more than $2$ stable states. The following $176$ stable states  have been detected using MAPLE 11 (to shorten the length of the list, we have converted the binary sequences to integers):
\par{\tiny
832, 836, 840, 844, 14208, 14212, 14216, 14220, 15233, 15234, 15235, 15237, 15238, 
15239, 15241, 15242, 15243, 15245, 15246, 15247, 245776, 245777, 245778, 245779, 245780, 245781,  
245782, 245783, 245784, 245785, 245786, 245787, 245788, 245789, 245790, 245791, 245792,  
245793, 245794, 245795, 245796, 245797, 245798, 245799, 245800, 245801, 245802, 245803,  
245804, 245805, 245806, 245807, 245808, 245809, 245810, 245811, 245812, 245813, 245814,  
245815, 245816, 245817, 245818, 245819, 245820, 245821, 245822, 245823, 250896, 250912,  
250928, 251921, 251922, 251923, 251937, 251938, 251939, 251953, 251954, 251955, 800640,  
800644, 800648, 800652, 801665, 801666, 801667, 801669, 801670, 801671, 801673, 801674,  
801675, 801677, 801678, 801679, 1049408, 1049412, 1049416, 1049420, 1049424, 1049428,  
1049432, 1049436, 1049440, 1049444, 1049448, 1049452, 1049456, 1049460, 1049464,  
1049468, 1849216, 1849220, 1849224, 1849228, 1849232, 1849236, 1849240, 1849244,  
1849248, 1849252, 1849256, 1849260, 1849264, 1849268, 1849272, 1849276, 1850241,  
1850242, 1850243, 1850245, 1850246, 1850247, 1850249, 1850250, 1850251, 1850253,  
1850254, 1850255, 1850257, 1850258, 1850259, 1850261, 1850262, 1850263, 1850265,  
1850266, 1850267, 1850269, 1850270, 1850271, 1850273, 1850274, 1850275, 1850277,  
1850278, 1850279, 1850281, 1850282, 1850283, 1850285, 1850286, 1850287, 1850289,  
1850290, 1850291, 1850293, 1850294, 1850295, 1850297, 1850298, 1850299, 1850301,  
1850302, 1850303.  }
\par
\medskip
\subsection*{The Boolean network of Zhang {\it et~al.} (2008)}
\par
This is the Boolean network given by the diagram of Fig. 2B in Zhang {\it et~al.} (2008). We introduce the variables as follws:
\begin{table}[h]
  \begin{center}
  {\renewcommand{\arraystretch}{1.4}\small 
  \begin{tabular}{|c|c|c|c|c|c|c|c|c|}
    \hline
   IL15	&  RAS	& ERK	&    JAK	& IL2RBT	& STAT3	& IFNGT	& FasL\\ \hline 	
   $x_1$ & $x_2$ & $x_3$ & $x_4$ & $x_5$	& $x_6$	& $x_7$	& $x_8$\\ \hline\hline

   PDGF  &  PDGFR &   PI3K	 &   IL2	&   BcIxL	  &       TPL2	  &  SPHK &	S1P	  \\ \hline	
   $x_9$ & $x_{10}$ & $x_{11}$ & $x_{12}$ & $x_{13}$ & $x_{14}$ & $x_{15}$ & $x_{16}$ \\ \hline\hline

   sFas   & Fas &  DISC  &  Caspase  &  Apoptosis	&\cellcolor{Mygrey} &\cellcolor{Mygrey} &\cellcolor{Mygrey}\\ \hline
   $x_{17}$ & $x_{18}$ & $x_{19}$ & $x_{20}$ & $x_{21}$ & \cellcolor{Mygrey} &\cellcolor{Mygrey} &\cellcolor{Mygrey}\\ \hline\hline
   LCK	 &  MEK	  &  GZMB	&  IL2RAT	&   FasT	&  RANTES	&  A20	&   FLIP\\	
  $x_{22}$ & $x_{23}$ & $x_{24}$ & $x_{25}$ & $x_{26}$ & $x_{27}$ & $x_{28}$ & $x_{29}$\\\hline 
 \end{tabular}
  }
  \label{tab: pairing2}
  \vskip 0.5cm
  \caption[Legend of variable names.]{\footnotesize Legend of variable names.}
    \end{center}
\end{table} 
\par
Then the polynomial representation of the Boolean network is given by
\begin{align*}
f_1 &= f_2 = f_4 = f_5 = f_{22} = x_1,\\
f_3 &= f_{23} = x_2,\\
f_6 &= f_{24} = x_4,\\
f_7 &= x_5 + x_6 + x_5x_6,\\
f_8 &= x_6(x_3 + x_5 + x_3x_5) + x_{14} + x_6(x_3 + x_5 + x_3x_5)x_{14},\\
f_9 &= f_{10} = x_9,\\
f_{11} &= x_{10},\\
f_{12} &= f_{13} = x_4 + x_{11} + x_4x_{11} +1,\\
f_{14} &= f_{29} = x_{11},\\
f_{15} &= x_{11} + x_{16} + x_{11}x_{16},\\
f_{16} &= f_{17} = x_{15},\\
f_{18} &= x_{17} + 1 + (x_1 + 1)(x_{11} + 1) + (x_{17} + 1)(x_1 + 1)(x_{11} + 1),\\
f_{19} &= x_{18},\\
f_{20} &= (x_1 + 1)x_{19},\\
f_{21} &= x_{20},\\
f_{25} &= x_{12},\\
f_{26} &= f_{27}= f_{28} = x_{14}.
\end{align*}
This Boolean network has more than 2 stable states. To find the fixed points, we can first ignore the last 8 Boolean functions $f_{22},\ldots, f_{29}$, since these functions correspond to some terminal nodes that no other node depends on them. After finding the fixed points for the system defined by the first $21$ functions, one can recover the fixed points for the Boolean network easily. For example, if $a = (a_1,a_2,\ldots, a_{29})$ is a fixed points, then $(a_1,\ldots, a_{21})$ must be a fixed point for the first $21$ functions, since they do not depend on the variables $x_{22},\ldots, x_{29}$. Also, by $f_{22}=x_1$, we must have $a_{22}=a_1$, and by $f_{23}=x_2$ we must have $a_{23}=a_2$, etc. The following 6 fixed points were dected by MAPLE:
\par
\medskip
{\footnotesize
00000000000110000111100010000,\;\;
00000000000110111111100010000 
\par   
00000001111001111000000001111,\;\;
11111111000000000110011100000
\par
11111111000000111000011100000,\;\;
11111111111001111000011101111.
}
\par
\medskip
\end{document}